\documentclass[12pt]{article}
\usepackage{setspace}
\usepackage{amsmath,amssymb}
\usepackage[T1]{fontenc}
\usepackage[bf,small,tableposition=top]{caption}
\usepackage{subfig}
\usepackage{amsthm}
\usepackage{lineno}
\usepackage[figuresright]{rotating}
\usepackage{enumerate}
\usepackage{authblk}
\usepackage[utf8]{inputenc}
\usepackage{graphics}
\usepackage{graphicx}
\usepackage{multicol}
\usepackage{lscape}
\usepackage[dvipsnames]{xcolor}
\usepackage{mathrsfs}
\usepackage{hyperref}

\usepackage{newunicodechar}
\newunicodechar{̧}{\c{}}

\usepackage{float}

\usepackage{verbatim}
\usepackage{diagbox}
\usepackage{multirow,bigdelim,dcolumn,booktabs}

\newtheorem{thm}{Theorem}
\newtheorem{lem}{Lemma}

\theoremstyle{definition}

\theoremstyle{remark}
\newtheorem{cor}{Corollary}
\newtheorem{rmk}{Remark}

\newcolumntype{H}{>{\setbox0=\hbox\bgroup}c<{\egroup}@{}}
\usepackage[authoryear]{natbib}
\title{A Stein Characterization-type Omnibus Tests for the Discrete Pareto Distribution}

\author
{Deepesh Bhati$^{1}$, Bruno Ebner$^{2}$, Sakshi Khandelwal$^{1}$\footnote{corresponding author: GSAKSHI1506@gmail.com} \\
\normalsize{$^{1}$Department of Statistics, Central University of Rajasthan, Ajmer, India \\
$^2$Institute of Stochastics, Karlsruhe Institute of Technology (KIT), Karlsruhe, Germany.}
}

\date{}

\begin{document} 


\baselineskip13pt

\maketitle 


\begin{abstract}
The discrete Pareto (or Zeta, Zipf) distribution, arises naturally in modeling rank-frequency data across diverse fields such as linguistics, demography, biology, and computer science. Despite its widespread applicability, goodness-of-fit testing for the discrete Pareto distribution remains underdeveloped, particularly in the presence of heavy tails and infinite support. This article introduces a novel goodness-of-fit test based on a new Stein-type characterization of the discrete Pareto distribution, formulated using its probability generating function. The proposed method is applicable even when the shape parameter is unknown and avoids binning or smoothing techniques. We study the asymptotic properties of the test and assess its empirical size and power through extensive simulation experiments. The results show that the proposed test either outperforms or matches the performance of existing method across various alternatives. Applications to real datasets are provided to demonstrate its practical relevance and robustness.
\end{abstract}


\onehalfspacing

\section{Introduction} 
Power laws describe statistical relationships where the probability of an event varies inversely with its magnitude. Specifically, a random variable $X$ is said to follow a power law if its probability mass function satisfies
\begin{equation*}\label{eq:pmf}
p_\nu(k) \propto k^{-\nu}, \quad k\in\mathbb{N}, \nu > 1,
\end{equation*}
where $\nu$ is known as the scaling parameter or exponent. Such distributions are heavy-tailed, meaning that extreme values, while rare, have non-negligible probability. Power-law distributions come in two basic forms: continuous and discrete. Continuous power laws are used when the variable of interest takes real values, while discrete power-law distributions apply when the variable takes values on a countable support, typically the set of natural numbers. Although continuous power laws, such as the classical Pareto distribution, have been studied extensively, discrete power laws pose distinct challenges in terms of modeling, inference, and goodness-of-fit (gof) testing due to their irregular spacing, infinite support, and data sparsity.\\
\indent In this paper, we focus on the discrete Pareto distribution $P_\nu$, with shape parameter $\nu>1$. This distribution has the probability mass function (pmf)
\begin{equation} \label{pmf}
p_{\nu}(k) = \frac{k^{-\nu}}{\zeta(\nu)}, \quad k \in \mathbb{N},\; \nu > 1,   
\end{equation}
\noindent and probability generating function (pgf) defined as $G_X(s)= \mathbb{E}[s^X] = \text{Li}_\nu(s)/\zeta(\nu)$, where $\zeta(\nu) = \sum_{k=1}^{\infty} k^{-\nu}$ and $\text{Li}_{\nu}(s)=\sum_{k=1}^{\infty} s^kk^{-\nu}$ are the Riemann zeta function and polylogarithm function respectively. Henceforth we denote the discrete Pareto Distribution as DPareto distribution, 
and denote the DPareto family of distributions having pmf $p_{\nu}$ given in \eqref{pmf} by $\mathcal{P}=\{ \mathrm{DPareto}(\nu): \nu>1\}$. Note that the moments of the DPareto law only exist for $\nu>2$ (first) and $\nu>3$ (second), so we potentially have a heavy tailed discrete distribution. The DPareto distribution is often used to model rank-frequency data, where a small number of items occur very frequently, and the rest appear rarely. This structure makes it particularly suitable for modeling phenomena with heavy-tailed, skewed frequency distributions. \\
\indent Power-law behavior is pervasive in natural and social systems, appearing in phenomena such as earthquake magnitudes, word frequencies, city populations, financial returns, internet traffic, and biological processes  \citep{newman2005power, mitzenmacher2004brief, clauset2009power}. The earliest known observation of power-law behavior dates back to Pareto’s work on income distribution \citep{pareto1896cours}. In the discrete case, \citet{zipf1949human} is a cornerstone example. Zipf observed that in written English texts, a small number of words are used extremely frequently, while the majority of words occur only occasionally. Since then, numerous studies \citep{baayen2001word, piantadosi2014zipf,hatzigeorgiu2001word} have shown that word frequency distributions in various languages and texts closely follow Zipf's law. Beyond linguistics, Zipf-like distributions have been observed in a variety of fields. \citet{moura2006zipf} showed that the population sizes of Brazilian cities conform to Zipf’s law, while \citet{gan2006zipf} analyzed U.S. and Chinese urban data and reported similar findings. In business contexts, revenues of 500 Chinese companies were shown to follow a Zipf-like structure. \citet{wang2017zipf} applied the distribution to analyze the strength of 13.3 million real-world passwords. More recently, \citet{naryzhny2023quantitative} used Zipf’s law to describe the abundance distribution of proteins and proteoforms in human cells. Owing to its wide applicability, the DPareto distribution has become a cornerstone model in linguistics \citep{baayen2001word}, demography \citep{gan2006zipf}, geography \citep{simon1955class}, biology \citep{mantegna1994linguistic}, climatology \citep{primo2007statistical}, and physics \citep{shreider1967theoretical}, to name a few examples.

\indent Despite its widespread application, statistical inference, especially gof testing for the DPareto distribution remains limited. Let $X_1,\ldots,X_n$ be independent and identically distributed (iid) copies of a random variable $X$ taking values in $\mathbb{N}$ and denote the distribution of $X$ by $\mathbb{P}^X$. The testing problem of interest is to test the composite hypothesis
\begin{equation}\label{eq:H0}
    H_0:\,\mathbb{P}^X\in\mathcal{P}
\end{equation}
against general alternatives based on the sample $X_1,\ldots,X_n$.
Traditional gof tests such as the Pearson chi-square and Kolmogorov–Smirnov (KS) tests are often inadequate in this context. These methods typically require binning, depend heavily on expected cell counts, and may lack power when applied to discrete, heavy-tailed data with infinite support  \citep{read1988goodness}. Moreover, they do not leverage the specific structure of the DPareto model, making them less effective in detecting meaningful departures from the null hypothesis.\\
\indent To address these limitations, recent research has turned to characterization-based testing, where unique properties of a distribution are inverted to form testable conditions. Among these approaches, Stein’s method \citep{stein1972bound} has emerged as a powerful framework for both approximation and testing. Originally developed for the normal distribution, Stein’s method has since been extended to a variety of continuous and discrete distributions \citep{ley2017stein}. At its core, Stein’s method provides distribution-specific identities, so called Stein characterizations, that can be transformed into gof tests \citep[Section 5.2]{Anastasiou2023}.\\
\indent In parallel, the pgf has proven to be an essential tool for studying discrete distributions. It encapsulates the entire distribution and is particularly well-suited for deriving functional identities and moments. While pgf-based approaches have been used in distribution theory and estimation, their potential in gof testing, especially in conjunction with Stein’s method has not been fully explored for discrete power law models.\\
\indent In this paper, we propose a novel gof test for the DPareto distribution, leveraging a Stein-type characterization combined with its pgf. We derive a new identity that uniquely characterizes the DPareto distribution and construct a test statistic based on the empirical deviation from this identity. The proposed test avoids arbitrary binning or smoothing and is applicable even when the scaling parameter $\nu$ is unknown. We also apply the test to real-world datasets to demonstrate its practical utility and robustness.\\
\indent The rest of the paper is structured as follows.
Section \ref{newch} presents a new characterization based on Stein's identity for discrete distributions. Section \ref{TSConstruction} describes the construction of new test statistic, and its limiting properties are discussed in Section \ref{limit}. Section \ref{MCstudy} presents Monte Carlo simulation results comparing the proposed method with existing gof test of DPareto distribution. Section \ref{data} provides applications to real-world datasets. Finally, Section \ref{conclusion} concludes with a discussion and possible directions for future research.

\section{A new Stein-type characterization of DPareto distribution} \label{newch}

The use of Stein-type characterizations in conjunction with the pgf provides a powerful framework for constructing gof tests for discrete distributions. Following the discrete Stein identity developed in \citet{ley2011discrete}, a random variable $X$ with support on $\mathbb{N}$ follows a distribution with pmf $p(x)$ if and only if
\[
\mathbb{E}\left( \Delta^+ f(X) + \frac{\Delta^+ p(X)}{p(X)} f(X+1) \right) = 0,
\]
for all test functions $f$ in a suitable function class, where $\Delta^+ f(x) = f(x+1) - f(x)$ is the forward difference operator. Consider $f_s(x) = 1-s^{x-1}$,  $s\in(0,1)$, which ensure $f_s(1)=0$, the above Stein's identity reduces to:
\begin{equation*}
    \mathbb{E} \bigg( (1-s^{X})-(1-s^{X-1}) + \frac{\Delta^+ p(X)}{p(X)} (1-s^{X}) \bigg) = 0, 
\end{equation*}
which leads after elementary calculations to the identity
\begin{equation*}
    (1-s) \, G_X(s) = \mathbb{E} \bigg( - \frac{\Delta^+ p(X)}{p(X)} \, s \, (1-s^{X})\bigg),\quad s\in(0,1).
\end{equation*}
 On the basis of above integration of Stein's characterization with pgf, we construct a new characterization to develop gof test for discrete distributions which is stated as follows. 

\begin{thm} \label{gchth}
Let $X$ be a discrete random variable taking values in $\mathbb{N}$ with pmf $p(x)$, and let $G_X(s)$ denote its pgf. Then $p(x)$ is the pmf of the target distribution if and only if
\begin{equation} \label{ch11}
(1 - s) \, G_X(s) = \mathbb{E} \left( - \frac{\Delta^+ p(X)}{p(X)} \cdot s(1 - s^X) \right), \quad s\in(0,1). 
\end{equation}
\end{thm}

\begin{rmk}
The equation (\ref{ch11}) stated in Theorem 1 may also be derived by evaluating the pgf of the rv through Stein’s characterization for discrete distributions, as developed in \cite{betsch2022characterizations}.   
\end{rmk}

\noindent Equation (\ref{ch11}) yields a new characterization result for the DPareto distribution, stated as follows.
\begin{thm} \label{chth}
Let $X$ be a random variable taking values in $\mathbb{N}$, with pmf $p(x)$, and let $G_X(s)$ denote its pgf. Then $p(x)$ is the pmf of the DPareto distribution with exponent $\nu>1$ if and only if
\begin{equation*}\label{eq:NullStein}
(1 - s) \, G_X(s) = \mathbb{E} \bigg( \bigg(1- \bigg(\frac{X}{X+1}\bigg)^\nu \bigg) \, s \, (1-s^X)\bigg),\quad s\in(0,1).
\end{equation*}
\end{thm}
\noindent After further straightforward simplification we get the following characterization.
\begin{cor}\label{cor:char}
The random variable $X$ follows the DPareto($\nu$)-law for some $\nu>1$, if and only if
\begin{equation}\label{eq:Stein_id}
    \mathbb{E}\left(s^{X-1}-1+\left(\frac{X}{X+1}\right)^\nu(1-s^X)\right)=0,\quad s\in(0,1).
\end{equation}
\end{cor}
\noindent By Corollary \ref{cor:char} a $L^2$-type departure measure can be constructed by
\begin{equation} \label{eq:DM}
\Delta_\nu = \int_0^1 \bigg( \mathbb{E}\left(s^{X-1}-1+\left(\frac{X}{X+1}\right)^\nu(1-s^X)\right) \bigg)^2 \, ds, \quad \nu>1.
\end{equation}
Therefore, based on the characterization, $\Delta_\nu=0$ if and only if the $X\sim\text{DPareto}(\nu)$, $\nu>1$. However, the quantities in \eqref{eq:DM} are unknown, and a test should be constructed based on the empirical counterpart of them.

\section{Test statistics construction} \label{TSConstruction}

In this section we propose a new testing procedure for the testing problem in \eqref{eq:H0} by replacing the expectation in the departure measure $\Delta_\nu$ in \eqref{eq:DM} by it's empirical counterpart and by plugging in a consistent estimator $\widehat{\nu}_n$ for the unknown shape parameter $\nu>1$.

\noindent A suitable choice for the estimator $\widehat{\nu}_n$ is the MLE of $\nu$. The MLE $\widehat{\nu}_n$ of $\nu$ satisfies the equation 
\begin{equation*}
    \frac{\zeta'(\nu)}{\zeta(\nu)} + \frac{1}{n} \sum_{i=1}^n \ln(x_i)=0.
\end{equation*}
It is not difficult to obtain an acceptable solution of this equation numerically. 

\noindent We propose the following $L^2$-type test statistic,
\begin{align*}
\mathfrak{K} =& n \int_0^1 \left( \frac{1}{n} \sum_{j=1}^n \left( s^{X_j-1}-1+\left(\frac{X_j}{X_j+1}\right)^{\widehat{\nu}_n}\left(1-s^{X_j}\right)\right)\right)^2 \, ds, 
\end{align*}
which quantifies the squared deviation from the Stein identity integrated over the pgf parameter domain $s \in (0,1)$. A large value of $\mathfrak{K}$ indicates deviation from the hypothesized DPareto distribution and leads to the rejection of the null hypothesis.

\noindent Straightforward evaluation of the integrals lead to the integration free representation
\begin{eqnarray} \label{TS1}
\mathfrak{K}= n&+&\frac{1}{n} \sum_{i,j=1}^n\left( \frac{1}{X_{i}+X_{j}-1}-\left(\frac1{X_{i}}+\frac1{X_{j}}\right) \right. \\ \nonumber&&\left. +2\left(\frac{X_{i}}{X_{i}+1}\right)^{\widehat{\nu}_n+1}\left(\frac{X_{i}+1}{X_{j}(X_{j}+X_{i})}-1\right) \right.
\\ \nonumber && \left.+ \left(\frac{X_{i}}{X_{i}+1}\right)^{\widehat{\nu}_n+1}\left(\frac{X_{j}}{X_{j}+1}\right)^{\widehat{\nu}_n+1}\frac{X_{i}+X_{j}+2}{X_{i}+X_{j}+1}\right). 
\end{eqnarray}
which is suitable for implementations.


\section{Limiting property of the test statistics \texorpdfstring{$\mathfrak{K}$}{K}} \label{limit}
Let $X_{1}, \ldots, X_{n},\ldots$ be iid random variables distributed according to the DPareto-law with parameter $\nu_0>1$. Assume further that the consistent estimator $\widehat{\nu}_n$ admits a linear representation of the form
\begin{equation}\label{eq:linear-rep}
    \sqrt{n}\left(\widehat{\nu}_n-\nu_0\right)=\frac1{\sqrt{n}}\sum_{j=1}^n\ell(X_{j},\nu_0)+o_{\mathbb{P}}(1),
\end{equation}
where $\ell$ is a known function satisfying $\mathbb{E}(\ell(X_{1},\nu_0))=0$ and $\mathbb{E}(\ell(X_{1},\nu_0)^2)<\infty$. Note that for choosing $\widehat{\nu}_n$ to be the MLE we have by the Bahadur representation, see Corollary 10.16 in \cite{H:2024}, $\ell(x,\nu_0)=-I(\nu_0)^{-1}\left(\log(x)+\zeta'(\nu_0)/\zeta(\nu_0)\right).$ 

\noindent For $\nu>1$ define
\begin{equation*}\label{eq:def-g}
g_\nu(x,s)
=s^{x-1}-1+\left(\frac{x}{x+1}\right)^\nu(1-s^x), \quad x\in\mathbb{N},s\in(0,1),
\end{equation*}
and the symmetric kernel
\begin{eqnarray*}\label{eq:def-h}
h_\nu(x,y)&=&\int_0^1 g_\nu(x,s)\,g_\nu(y,s)\,ds\\
&=&\frac{1}{x+y-1}-\left(\frac1{x}+\frac1{y}\right)\nonumber \\ &&+\left(\frac{x}{x+1}\right)^{\nu+1}\left(\frac{x+1}{y(y+x)}-1\right)+\left(\frac{y}{y+1}\right)^{\nu+1}\left(\frac{y+1}{x(y+x)}-1\right) \nonumber \\ &&+\left(\frac{x}{x+1}\right)^{\nu+1}\left(\frac{y}{y+1}\right)^{\nu+1}\frac{x+y+2}{x+y+1}+1 ,\qquad x,y\in\mathbb{N}.
\end{eqnarray*}
The statistic $\mathfrak{K}$ given in \eqref{TS1} is a V-statistic of degree 2
\begin{equation*}\label{eq:K-Vstat-main}
\mathfrak{K}
= \frac{1}{n}\sum_{i,j=1}^n h_{\widehat{\nu}_n}(X_{i},X_{j})
= n\,V_n(h_{\widehat{\nu}_n}),
\qquad 
V_n(h)=\frac{1}{n^2}\sum_{i,j=1}^n h(X_{i},X_{j}),
\end{equation*}
To account for parameter estimation, we introduce the functions
\[
\psi_{\nu_0}(x,s)=g_{\nu_0}(x,s)+a(s)\,\ell(x,\nu_0), \quad x\in\mathbb{N},
\qquad 
a(s)=\mathbb{E}_{\nu_0}\!\left[\partial_\nu g_{\nu_0}(X,s)\right],\quad s\in(0,1),
\]
and the induced kernel
\[
k_{\nu_0}(x,y)=\int_0^1 \psi_{\nu_0}(x,s)\psi_{\nu_0}(y,s)\,ds,\quad x,y\in\mathbb{N}.
\]
Under the Stein identity in \eqref{eq:Stein_id}  we have $\mathbb{E}_{\nu_0}[g_{\nu_0}(X,s)]=0$ for all $s\in(0,1)$, such that the kernel $k$ is 
$1$-degenerate. Moreover, $\mathfrak{K}=nV_n(k)+o_{\mathbb{P}}(1)$ and $k\in L^2(P_{\nu_0}\times P_{\nu_0})$,
see Lemma \ref{lem:kernel-regularity-app} in Appendix~\ref{app:Vstat-proofs}. Here, $L^2(P_\nu)=\{f:\mathbb{N}\to\mathbb{R}: \sum_{k=1}^\infty (f(k))^2p_{\nu}(k)<\infty\}$ equipped with the inner product $\langle f,g\rangle=\sum_{k=1}^\infty f(k)g(k)p_{\nu}(k)$ denotes the space of square-integrable functions with respect to the distribution $P_\nu$.

\begin{thm}\label{thm:null-limit}
Let $(\lambda_m)_{m\ge 1}$ be the positive decreasing sequence of eigenvalues of the Hilbert--Schmidt operator
$T:L^2(P_{\nu_0})\to L^2(P_{\nu_0})$ defined by
\[
(Tf)(x):=\int_{\mathbb{N}} k_{\nu_0}(x,y)f(y)\,dP_{\nu_0}(y)=\sum_{y=1}^\infty k_{\nu_0}(x,y)f(y)p_{\nu_0}(y).
\]
Then, under the standing assumptions of this section,
\begin{equation*}\label{eq:limit-law}
\mathfrak{K}\ \stackrel{d}{\rightarrow}\ \sum_{m=1}^\infty \lambda_m Z_m^2,\quad\mbox{as}\;n\rightarrow\infty,
\end{equation*}
where $Z_1,Z_2,\dots$ are iid  $N(0,1)$ random variables.
\end{thm}
\noindent The proof of Theorem~\ref{thm:null-limit} relies on the results of \citet{LN:2013} and is given in Appendix~\ref{app:Vstat-proofs}. We now prove that the test which rejects the hypothesis $H_0$ for large values of $\mathfrak{K}$ is consistent against general alternatives. Hereafter, we consider an iid sequence $(X_n)_{n\in \mathbb{N}}$ of copies of $X$, where  $X$ is a non-degenerate positive random variable taking values in $\mathbb{N}$. Moreover, we assume that there is $\nu_0>1$ such that
\begin{align*}
\widehat{\nu}_{n} \stackrel{\text{a.s.}}{\longrightarrow} \nu_0, \qquad \text{as} \quad n \rightarrow \infty,
\label{consistency_convergence_estimators}
\end{align*}
where \(\widehat{\nu}_{n}\) is the consistent estimator as before. The following result is a direct consequence of a Taylor expansion and Fatou's lemma.
\begin{thm}\label{thm:consistency}
Under the stated conditions, we have
\begin{align*}
\liminf_{n \rightarrow \infty} \frac{\mathfrak{K}}{n} \geq \Delta_{\nu_0} \qquad \mathbb{P} \textnormal{-a.s.},
\end{align*}
where $\Delta_{\nu_0}$ is the departure measure defined in \eqref{eq:DM}.
\end{thm}
Note that Theorem \ref{thm:consistency} implies global consistency of the testing procedure since $\Delta_{\nu_0}>0$ whenever the underlying law is not a member of the discrete Pareto family $\mathcal{P}$, which is directly implied by the characterization, see Corollary \ref{cor:char}.

\noindent Since the eigenvalues of the covariance operator in the limit null distribution of $\mathfrak{K}$ in Theorem \ref{thm:null-limit} depend on the unknown parameter $\nu_0>1$ of the underlying 
DPareto distribution, we propose a parametric bootstrap procedure to obtain critical values. 
For a sample $X_1,\ldots,X_n$ satisfying the assumptions above, we compute the value of the consistent estimator $\widehat{\nu}_n=\widehat{\nu}_n(X_1,\ldots,X_n)$.

\noindent We then generate a bootstrap sample of size $n$, say $X_1^*,\ldots,X_n^*$, following the 
$\mathrm{DPareto}(\widehat{\nu}_n)$ distribution, estimate the parameter $\nu$ from $X_1^*,\ldots,X_n^*$ 
(denote it $\widehat{\nu}_n^*$), and calculate the test statistic $\mathfrak{K}^*$. By repeating this procedure $b$ times, we obtain $\mathfrak{K}_{n,1}^*,\ldots,\mathfrak{K}_{n,b}^*$ and compute 
the empirical distribution function
\[
H_{n,b}^*(t)=\frac{1}{b}\sum_{i=1}^b\mathbb{I}(\mathfrak{K}_{n,i}^*\le t),\qquad t\ge 0,
\]
where $\mathbb{I}$ stands for the indicator function. Given the nominal level $\alpha\in[0,1]$, we use the empirical $(1-\alpha)$-quantile:
\[
c_{n,b}^*(\alpha)=H_{n,b}^{*-1}(1-\alpha)=\begin{cases}
\mathfrak{K}_{b(1-\alpha):b}^*, & b(1-\alpha)\in\mathbb{N},\\[0.2em]
\mathfrak{K}_{\lfloor b(1-\alpha)\rfloor+1:b}^*, & \text{otherwise,}
\end{cases}
\]
where $\mathfrak{K}_{1:b}^*,\ldots,\mathfrak{K}_{b:b}^*$ are the order statistics. We reject $H_0$ if $\mathfrak{K}>c_{n,b}^*(\alpha)$.

\noindent Denote the distribution function of $\mathfrak{K}$ under $\mathrm{DPareto}(\nu)$ by 
$H_{n,\nu}(t)=\mathbb{P}_\nu(\mathfrak{K}\le t)$, and the limit distribution by 
$H_\nu(t)=\mathbb{P}( \sum_{m=1}^\infty \lambda_m Z_m^2\le t)$, where $ \sum_{m=1}^\infty \lambda_m Z_m^2$ is the random variable from the limit null distribution.  The function $H_\nu$ is continuous and strictly monotone. By consistency of $\widehat{\nu}_n$ 
and continuity of $H_\nu$, for each $t\ge 0$,
\[
\lim_{n\to\infty}H_{n,\widehat{\nu}_n}(t)=H_{\nu_0}(t)\quad\mathbb{P}\text{-a.s.}
\]

\noindent By a triangular version of the bootstrap argument  \citep[Theorem 3.6]{H:1996}, we have
\[
\sup_{t\ge 0}\left|H_{n,b}^*(t)-H_{n,\widehat{\nu}_n}(t)\right|\xrightarrow{\mathbb{P}}0
\quad\text{as }b,n\to\infty.
\]
Thus, $c_{n,b}^*(\alpha)\xrightarrow{\mathbb{P}}H_{n,\widehat{\nu}_n}^{-1}(1-\alpha)$ as $b\to\infty$. If $X_1,\ldots,X_n$ are iid\ from $\mathrm{DPareto}(\nu_0)$, the continuity of $H_{\nu_0}$ yields
\[
\lim_{n\to\infty}\lim_{b\to\infty}\mathbb{P}(\mathfrak{K}>c_{n,b}^*(\alpha))=\alpha.
\] 
If $X_1,\ldots,X_n$ do not follow a DPareto distribution, then by Corollary \ref{cor:char} and Equation \eqref{eq:DM}, 
$\Delta_{\nu_0}>0$, so
\[
\lim_{n\to\infty}\lim_{b\to\infty}\mathbb{P}(\mathfrak{K}>c_{n,b}^*(\alpha))=1.
\]
Thus the test is consistent against any fixed alternative distribution satisfying 
the stated assumptions.

\begin{rmk}
    The derived results could also be obtained by taking advantage of the Hilbert-space structure of the underlying $L^2$-space and deriving the weak limit of the empirical Stein process 
    \begin{equation*}
    Z_n(s)=\frac1{\sqrt{n}}\sum_{j=1}^ng_{\widehat{\nu}_n}(X_j,s),\quad s\in(0,1).
    \end{equation*}
    This approach would follow the line of Section 2 in \cite{EH:2025} by considering the fact $\mathfrak{K}=\|Z_n(X_j,\cdot)\|_{L^2}^2$ and that the limit distribution in Theorem \ref{thm:null-limit} is the Karhunen-Loève expansion of the resulting limit centred Gaussian process.
\end{rmk}

\section{Simulation Study} \label{MCstudy}
The properties discussed in the previous sections are asymptotic, meaning that they describe the behavior of the proposed test when the sample size is large. However, for small and moderate sample sizes, we perform a Monte Carlo simulation with $MC=1000$ replications. In each replication, we estimate $\hat{\nu}$ and generate $b=500$ bootstrap samples of size $n$ from DPareto$(\hat{\nu})$ and obtain the bootstrap critical point. Then the rejection probabilities are estimated by obtaining the proportion of rejection of the null hypothesis in 1000 Monte Carlo replications. For comparison,
we include the test $Z_{n,a}$ proposed by \citet{meintanis2009unified} and three recently used gof tests viz. $T_{n,\beta}$, $C_n^e$ and $S_{\text{BEN}}$. The brief details of these tests are summarized as follows:
\begin{enumerate}
    \item[(i)] The Test $Z_{n,a}$  proposed by \citet{meintanis2009unified}, involves the difference between the empirical inverse Mellin transform and its theoretical counterpart, is given as
\begin{equation*} 
Z_{n,a} = \zeta^2(\nu) \frac{1}{n} \sum_{j,k=1}^n I_w^{(0)} (x_jx_k) + n I_w^{(2)} (1) - 2 \zeta(\nu) \sum_{j=1}^n I_w^{(1)} (x_j), 	
\end{equation*}
where	
\begin{equation} \nonumber
I_w^{(m)} (x)= \int_0^\infty \zeta^m(\nu+t) \frac{1}{x^t} w(t) dt, \quad m=0,1,2. 	 
\end{equation} 
By setting $w(t) = \exp(-at)$, they considered \[I_a^{(0)} (x) := (a+ \log(x))^{-1} , \]
\[ I_a^{(1)} (x) := I_a^{(1)} (x, \nu) = \sum_{s=1}^\infty \frac{1}{s^\nu (a+\log(sx))}, \] while the value of $I_a^{(2)} (1)$, which is independent of the observations obtained by numerical integration.
\item[(ii)] The Test $T_{n,\beta}$ proposed by \cite{EH:2025}, is based on Steins characterization of Zeta law using the generator approach, given as
\begin{align*}
T_{n,\beta} = \frac{2}{n} \sum_{j,k=1}^n &\bigg(\frac{\mathbb{I} (X_{n,j} = X_{n,k} = 1)}{(3+\beta) \, (4+\beta) \, (5+\beta)} \\ 
   &- 2 \bigg( \left(\frac{X_{n,j}}{X_{n,j}-1}\right)^{\widehat{\nu}_n}
 - \frac{X_{n,j}+1}{\beta + X_{n,j} + 4}   \bigg) \, B(X_{n,j}+1, 3+\beta) \, \mathbb{I} (X_{n,j} \geq 2, X_{n,k} = 1) \\
   &+ \bigg ( B(X_{jk}^+ -1, 3+\beta) \, \left(\frac{X_{n,j}}{X_{n,j}-1}\right)^{\widehat{\nu}_n}
 \left(\frac{X_{n,k}}{X_{n,k}-1}\right)^{\widehat{\nu}_n}
 \\ 
   & -  B(X_{jk}^+, 3+\beta) \, \left( \left(\frac{X_{n,j}}{X_{n,j}-1}\right)^{\widehat{\nu}_n}
 +\left(\frac{X_{n,k}}{X_{n,k}-1}\right)^{\widehat{\nu}_n}
 \right) \\ & + B(X_{jk}^+ +1, 3+\beta) \bigg) \, \mathbb{I} (X_{n,j}, X_{n,k} \geq 2)\bigg),
\end{align*}
where $B(\cdot,\cdot)$ denotes the Beta function and $X_{jk}^+ = X_{n,j} + X_{n,k}$.

\item[(iii)] Test $C_n^e$ by \citet{H:1996} for parametric families of discrete distributions is based on the difference between the empirical distribution function and the estimated distribution function.
In the present setting, we use its specialization to the DPareto distribution with unknown parameter $\nu$. The corresponding test statistic is defined as
\[
C_n^e
= n \sum_{k=1}^\infty
\big( \widehat F_n(k) - F(k;\widehat{\nu}_n) \big)^2
\big( \widehat F_n(k) - \widehat F_n(k-1) \big),
\]
where
\[
\widehat F_n(k) - \widehat F_n(k-1)
= \frac{1}{n} \sum_{j=1}^n \mathbb{I} (X_j = k), \quad k \geq 1,
\]
denotes the empirical pmf.

\item[(iv)] Finally, we consider the test proposed by \citet{betsch2022characterizations}, which is based on the discrepancy measure between the empirical Steins-type pmf identity and its empirical counterpart. Its adaptation to the DPareto distribution yields the test statistic
\[
S_{\mathrm{BEN}} (\widehat{\nu}_n)
= \sum_{k=1}^{M_n}
\big( e_n(k;\widehat{\nu}_n) - \rho_n(k) \big)^2,
\quad
M_n = \max_{1 \leq j \leq n} X_j ,
\]
where
\[
e_n(k;\widehat{\nu}_n)
= \frac{1}{n} \sum_{j=1}^n
\bigg( 1- \left( \frac{X_j}{X_j+1} \right)^{\widehat{\nu}_n} \bigg) \, 
\mathbb{I} (X_j \geq k),
\]
and
\[
\rho_n(k)
= \frac{1}{n} \sum_{j=1}^n \mathbb{I} (X_j = k).
\]
\end{enumerate}
\noindent To compare the power of the proposed test with existing gof procedures, we consider a range of alternative distributions supported on 
$\mathbb{N}$. These alternatives are constructed by perturbing the DPareto distribution in two distinct ways: (i) through the sum i.e. $X_1+X_2$ and (ii) through the maximum 
$\max(X_1,X_2)$. Here, $X_1$ follows a DPareto distribution with shape parameter $\nu$ and $X_2$ is an independent random variable supported on 
$\mathbb{N}$. Since the resulting random variables do not follow a DPareto distribution, they serve as meaningful alternatives for power assessment. Specifically, $X_2$ is assumed to follow a discrete uniform distribution on $\{0,1,2,\cdots,k\}$, with $k=2,4,$ and 5 denoted by $DU(2)$, $DU(4)$ and $DU(5)$ respectively. It is noted that as the value of $k$ increases, the probability $p_{X_2}(0)$ decreases, resulting in deviation from DPareto distribution. Further, we consider the shape parameter $\nu$ =1.5 (addressing the infinite first moment case) and $\nu$= 2, and 3 (for the finite moment case). The empirical size and power of the competing tests for sample sizes $n$=10 and $n$=20 are reported in Tables \ref{PS1} and \ref{PS2}.\\
\noindent It is evident that all the tests under consideration adequately preserve the nominal significance level, with empirical sizes remaining close to the target level of 0.05. Furthermore, for alternative distributions such as the Poisson, binomial, negative binomial, and Bell distributions, all procedures attain power values approaching unity even for a small sample size of $n=10$. Consequently, the results for these alternatives are omitted, as they do not provide meaningful discrimination among the competing tests.  
For alternatives constructed via additive perturbations of the form $X_1+X_2$, the proposed test $\mathfrak{K}$ demonstrates uniformly high power across all combinations of $\nu$ and for $k$. In particular, for moderate to large perturbations ($k$=4,5), $\mathfrak{K}$-test attains power values exceeding 0.90 in almost all cases, often approaching unity. While the $Z_{n,a}$ tests also show competitive performance —especially for smaller values of $a$ — their power declines noticeably as $a$ increases. In contrast, the $T_{n,a}$ tests, along with 
$S_{BEN}$ and $C_n^e$ exhibit substantially lower power under these additive alternatives, particularly when the discrete uniform component has a larger support. These findings suggest that $\mathfrak{K}$ test is especially sensitive to departures from the DPareto model induced by additive contamination.\\
\noindent A different pattern emerges for alternatives based on the maximum structure $\max (X_1,X_2)$. In this case, the $T_{n,a}$ tests exhibit the strongest power performance, particularly for smaller values of $k$ and larger $\nu$, where power often exceeds 0.90. The proposed  $\mathfrak{K}$ test and the 
$Z_{n,a}$ tests also show increasing power as $\nu$ and $k$ grow, but they are generally outperformed by the $T_{n,a}$ family for these max-type alternatives. The tests $S_{BEN}$ and $C_n^e$ perform competitively for larger $k$, occasionally matching or slightly exceeding $\mathfrak{K}$, but their power is less stable across the full range of alternatives.
Overall, the simulation study reveals that no single test uniformly dominates across all types of alternatives; however, the proposed  $\mathfrak{K}$ test exhibits superior and robust power against a broad class of additive perturbations, which are practically relevant in modeling deviations from DPareto behavior. When combined with its correct size and competitive performance under max-type alternatives, these results establish 
$\mathfrak{K}$-test as a strong and reliable addition to the existing gof tests for the DPareto distribution.

\begin{landscape} 
\begin{table}[t]
\centering
\caption{Empirical size and power comparison for sample size $n=10$ (in percentage).} \label{PS1} 
\scalebox{0.9}{
\begin{tabular}{l|rrrHrrrrrrrrrr} 
\hline 
Distributions      & {$\mathfrak{K}$} & ${Z}_{10,0.5}$ & ${Z}_{10,1.0}$ & ${Z}_{10,1.5}$ & ${Z}_{10,2}$ & $T_{10,0}$ & $T_{10,1}$ & $T_{10,2}$ & $T_{10,3}$ & $T_{10,4}$ & $T_{10,5}$ & $S_{BEN}$ & $C_n^e$ \\ \hline 
$X_1(1.5)$            & 6.3              & 5.1     & 4.6     & 4.2     & 3.9   & 6.4      & 6.5      & 6.5      & 6.3      & 6.3      & 6.3      & 5.4    & 5.2                     \\
$X_1(2)$              & 5.1              & 4.7     & 4.4     & 4.5     & 4.7   & 5.7      & 5.8      & 5.8      & 5.8      & 5.8      & 5.8      & 4.3    & 5.2                     \\
$X_1(3)$              & 5.0              & 4.1     & 4.1     & 4.0     & 4.2   & 3.4      & 3.4      & 3.4      & 3.4      & 3.4      & 3.4      & 4.4    & 5.0                     \\
$X_1(1.5)+DU(2)$      & \textbf{52.9}             & 40.3    & 24.4    & 8.7     & 1.3   & 20.5     & 19.3     & 18.2     & 17.8     & 17.4     & 17.4     & 17.4   & 17.4                    \\
$X_1(2)+DU(2)$      & \textbf{77.3}             & 61.7    & 51.1    & 33.1    & 15.4  & 38.9     & 37.3     & 35.2     & 33.8     & 33.1     & 33.1     & 33.1   & 33.1                    \\
$X_1(2.5)+DU(2)$      & \textbf{83.5}             & 71.3    & 63.9    & 47.0    & 23.2  & 48.9     & 46.9     & 45.5     & 44.1     & 42.9     & 42.9     & 42.9   & 42.9                    \\
$X_1(3)+DU(2)$       & \textbf{85.7}             & 69.8    & 64.2    & 49.6    & 27.8  & 53.6     & 52.2     & 50.8     & 48.7     & 47.8     & 47.8     & 47.8   & 47.8                    \\
$X_1(1.5)+DU(4)$      & \textbf{64.9}             & 63.5    & 46.6    & 17.7    & 2.9   & 8.2      & 7.4      & 6.9      & 6.9      & 6.9      & 6.9      & 6.9    & 6.9                     \\
$X_1(2)+DU(4)$       & \textbf{85.3}             & 75.5    & 66.9    & 46.9    & 13.8  & 19.1     & 17.8     & 17.0     & 16.4     & 16.2     & 15.6     & 77.1   & 73.4                    \\
$X_1(2.5)+DU(4)$       & \textbf{88.7}             & 80.9    & 73.6    & 55.5    & 19.8  & 26.5     & 24.4     & 23.1     & 22.0     & 21.8     & 21.3     & 77.3   & 81.4                    \\
$X_1(3)+DU(4)$         & \textbf{89.3}             & 79.7    & 74.7    & 53.4    & 23.0  & 32.0     & 29.1     & 27.4     & 26.2     & 25.4     & 24.9     & 78.0   & 83.9                    \\
$X_1(1.5)+DU(5)$       & 65.5             & 68.0    & 51.7    & 22.8    & 2.5   & 3.5      & 3.7      & 3.7      & 3.7      & 3.7      & 3.5      & \textbf{73.6}   & 47.4                    \\
$X_1(2)+DU(5)$      & \textbf{85.5}             & 80.8    & 72.3    & 48.4    & 10.3  & 11.3     & 10.4     & 10.0     & 9.5      & 9.5      & 9.5      & 81.2   & 79.3                    \\
$X_1(2.5)+DU(5)$      & \textbf{88.5}             & 84.3    & 77.1    & 58.9    & 19.4  & 18.1     & 17.0     & 16.1     & 15.5     & 14.6     & 14.5     & 82.2   & 86.3                    \\
$X_1(3)+DU(5)$        & \textbf{89.1}             & 83.1    & 75.2    & 57.6    & 19.8  & 23.6     & 22.1     & 20.5     & 19.6     & 19.2     & 19.0     & 82.8   & 88.0                    \\
$\max(X_1(1.5),DU(2))$ & 18.0             & 11.1    & 4.8     & 1.9     & 0.5   & 23.6     & 24.0     & 24.5     & 25.0     & 25.3     & \textbf{25.4}     & 22.3   & 5.8                     \\
$\max(X_1(2),DU(2))$ & 34.4             & 23.5    & 17.2    & 11.1    & 4.3   & \textbf{55.3}     & 55.2     & 55.0     & 55.0     & 55.0     & 55.2     & 22.0   & 22.7                    \\
$\max(X_1(2.5),DU(2))$ & 41.2             & 32.0    & 24.8    & 14.4    & 7.0   & 73.7     & 74.0     & 74.5     & 74.4     & 74.5     & \textbf{74.9}     & 19.9   & 32.6                    \\
$\max(X_1(3),DU(2))$   & 39.4             & 27.4    & 20.9    & 11.0    & 5.0   & 79.6     & 79.7     & 79.7     & \textbf{79.8}     & \textbf{79.8}     & \textbf{79.8}     & 17.8   & 31.0                    \\
$\max(X_1(1.5),DU(4))$ & 48.6             & 35.9    & 22.7    & 8.5     & 1.6   & 12.6     & 11.8     & 11.6     & 11.5     & 11.1     & 11.2     & \textbf{50.3}   & 21.4                    \\
$\max(X_1(2),DU(4))$   & \textbf{66.9}             & 50.5    & 42.7    & 26.6    & 8.9   & 35.2     & 33.9     & 32.6     & 32.1     & 31.4     & 30.8     & 53.2   & 52.0                    \\
$\max(X_1(2.5),DU(4))$ & \textbf{70.6}             & 52.5    & 44.3    & 29.1    & 12.3  & 42.5     & 40.4     & 39.0     & 38.3     & 37.8     & 37.6     & 52.1   & 58.3                    \\
$\max(X_1(3),DU(4))$   & \textbf{66.4}             & 49.6    & 42.6    & 27.1    & 9.1   & 41.4     & 40.1     & 39.2     & 38.7     & 38.5     & 37.7     & 47.2   & 57.4                    \\
$\max(X_1(1.5),DU(5))$ & \textbf{54.3}             & 44.5    & 31.0    & 10.2    & 1.7   & 9.9      & 8.9      & 8.5      & 8.3      & 8.1      & 8.1      & 8.1    & 8.1                     \\
$\max(X_1(2),DU(5))$ & \textbf{73.1}             & 61.4    & 49.5    & 29.9    & 10.0  & 16.9     & 15.4     & 14.8     & 14.0     & 13.5     & 13.5     & 13.5   & 13.5                    \\
$\max(X_1(2.5),DU(5))$ & \textbf{73.0}             & 57.7    & 50.5    & 32.5    & 11.7  & 17.5     & 16.3     & 15.5     & 15.0     & 14.7     & 14.7     & 14.7   & 14.7                    \\
$\max(X_1(3),DU(5))$   & \textbf{68.7}             & 53.3    & 47.4    & 28.0    & 9.3   & 18.0     & 16.6     & 16.2     & 16.1     & 15.9     & 15.9     & 15.9   & 15.9  \\
\hline 
\end{tabular} }
\end{table}

\begin{table}[t]
\centering
\caption{Empirical size and power comparison for sample size $n$=20 (in percentage).} \label{PS2} 
\scalebox{0.9}{
\begin{tabular}{l|rrrrrrrrrrrr}
\hline 
Distributions      & {$\mathfrak{K}$} & ${Z}_{10,0.5}$ & ${Z}_{10,1.0}$  & ${Z}_{10,2}$ & $T_{10,0}$ & $T_{10,1}$ & $T_{10,2}$ & $T_{10,3}$ & $T_{10,4}$ & $T_{10,5}$ & $S_{BEN}$ & $C_n^e$ \\ \hline 
$X_1(1.5)$            & 6.0              & 5.9     & 6.7     & 5.1   & 5.3      & 5.3      & 5.1      & 5.2      & 5.2      & 5.3      & 5.4    & 5.7                     \\
$X_1(2)$              & 5.5              & 4.9     & 4.4     & 4.9   & 4.3      & 4.4      & 4.4      & 4.3      & 4.3      & 4.4      & 4.6    & 5.1                     \\
$X_1(3)$ & 6.0  & 6.1     & 5.9     & 5.7   & 5.1      & 5.1      & 5.1      & 5.0      & 4.7      & 4.7      & 5.2    & 5.7  \\
$X_1(1.5)+DU(2)$       & \textbf{81.8}             & 74.9    & 61.4    & 13.4  & 43.9     & 41.6     & 39.8     & 37.7     & 36.7     & 36.7     & 36.7   & 36.7                    \\
$X_1(2)+DU(2)$       & \textbf{96.1}             & 93.4    & 90.2    & 67.3  & 70.9     & 67.9     & 65.5     & 63.0     & 61.1     & 61.1     & 61.1   & 61.1                    \\
$X_1(2.5)+DU(2)$       & \textbf{97.7}             & 96.3    & 95.5    & 88.2  & 79.7     & 77.7     & 75.9     & 75.1     & 73.9     & 73.9     & 73.9   & 73.9                    \\
$X_1(3)+DU(2)$         & \textbf{98.6}  & 97.2    & 96.7    & 93.5  & 83.9     & 81.5     & 80.0     & 77.6     & 76.7     & 76.7     & 76.7   & 76.7                    \\
$X_1(1.5)+DU(4)$       & 91.2     & \textbf{92.2}    & 85.9    & 27.9  & 15.9     & 13.4     & 12.1     & 11.0     & 10.5     & 10.5     & 10.5   & 10.5                    \\
$X_1(2)+DU(4)$      & \textbf{98.4}             & 98.1    & 96.8    & 80.1  & 36.1     & 32.0     & 29.1     & 27.3     & 26.0     & 26.0     & 26.0   & 26.0                    \\
$X_1(2.5)+DU(4)$      & \textbf{99.2}             & 99.1    & 98.5    & 91.9  & 43.6     & 38.9     & 36.6     & 34.1     & 32.0     & 32.0     & 32.0   & 32.0                    \\
$X_1(3)+DU(4)$        & \textbf{99.4}             & 98.8    & 98.7    & 94.7  & 48.6     & 44.0     & 40.0     & 38.0     & 36.2     & 36.2     & 36.2   & 36.2                    \\
$X_1(1.5)+DU(5)$      & 92.8             & \textbf{94.7}    & 90.5    & 40.2  & 8.1      & 6.9      & 6.1      & 5.4      & 5.2      & 5.2      & 5.2    & 5.2                     \\
$X_1(2)+DU(5)$      & \textbf{98.6}             & \textbf{98.6}    & 98.4    & 84.4  & 24.9     & 21.9     & 18.7     & 17.4     & 16.1     & 16.1     & 16.1   & 16.1                    \\
$X_1(2.5)+DU(5)$      & \textbf{99.4}             & 99.3    & 99.2    & 93.8  & 32.1     & 27.7     & 25.1     & 23.2     & 22.3     & 22.3     & 22.3   & 22.3                    \\
$X_1(3)+DU(5)$        & \textbf{99.4}            & 98.9    & 98.9    & 95.4  & 35.8     & 32.1     & 30.0     & 28.0     & 25.9     & 25.9     & 25.9   & 25.9                    \\
$\max(X_1(1.5),DU(2))$ & 31.9             & 24.2    & 14.1    & 1.7   & 42.6     & 42.4     & 42.7     & 43.0     & 43.1     & \textbf{43.2}     & 39.2   & 13.9                    \\
$\max(X_1(2),DU(2))$ & 57.8             & 47.6    & 40.6    & 23.3  & 65.1     & \textbf{65.3}     & 65.2     & \textbf{65.3}    & \textbf{65.3}    & \textbf{65.3}    & \textbf{65.3}   & \textbf{65.3}                   \\
$\max(X_1(2.5),DU(2))$ & 70.8             & 62.1    & 56.3    & 42.8  & 91.2     & 91.1     & 91.1     & \textbf{91.4}     & 91.3     & 91.3     & 50.5   & 62.6                    \\
$\max(X_1(3),DU(2))$   & 75.4             & 68.3    & 61.7    & 48.0  & \textbf{96.0}     & 95.9     & \textbf{96.0}     & \textbf{96.0}     & \textbf{96.0}    & \textbf{96.0}     & 47.8   & 63.5                    \\
$\max(X_1(1.5),DU(4))$ & 74.9             & 66.3    & 49.1    & 9.4   & 28.6     & 26.2     & 25.0     & 24.4     & 23.5     & 22.9     & \textbf{76.6 }  & 53.4                    \\
$\max(X_1(2),DU(4))$   & \textbf{89.7}             & 84.2    & 79.9    & 55.6  & 60.2     & 57.0     & 54.6     & 53.3     & 51.5     & 49.9     & 84.8   & 87.3                    \\
$\max(X_1(2.5),DU(4))$ & \textbf{93.7}    & 88.9    & 87.8    & 73.4  & 66.8     & 64.0     & 61.8     & 59.1     & 57.7     & 56.5     & 86.4   & 92.4                    \\
$\max(X_1(3),DU(4))$   & \textbf{93.7}             & 86.7    & 86.5    & 76.2  & 65.7     & 62.6     & 60.0     & 57.5     & 56.0     & 55.1     & 83.1   & 90.9                    \\
$\max(X_1(1.5),DU(5))$ & 82.7             & 78.8    & 67.1    & 14.6  & 19.7     & 18.0     & 16.9     & 16.3     & 15.2     & 14.8     & \textbf{84.7}   & 70.8                    \\
$\max(X_1(2),DU(5))$ & \textbf{94.4}             & 91.7    & 89.6    & 67.5  & 45.0     & 41.8     & 40.2     & 38.1     & 35.7     & 34.5     & 91.7   & \textbf{94.4}                   \\
$\max(X_1(2.5),DU(5))$ & \textbf{95.7}             & 92.8    & 92.5    & 80.5  & 51.2     & 48.0     & 44.3     & 42.5     & 41.6     & 40.6     & 90.9   & 95.6                    \\
$\max(X_1(3),DU(5))$   & \textbf{94.3}             & 90.1    & 89.8    & 79.7  & 50.0     & 45.7     & 43.9     & 42.2     & 40.9     & 40.0     & 89.1   & 94.0     \\ \hline                
\end{tabular} }
\end{table}
\end{landscape}

\section{Data Analysis} \label{data}

In this section, we illustrate the application of the proposed gof test for the DPareto distribution using the real dataset presented by \citet{van2003point}, which reports the number of times illegal immigrants were apprehended by the police in four major Dutch cities—Amsterdam, Rotterdam, The Hague, and Utrecht—during the year 1995. Based on the outcome of the apprehension process, individuals were categorized as: (i) effectively expelled, (ii) not effectively expelled, and (iii) other/missing. Our analysis focuses on the first two groups: those who were effectively expelled and those who were not effectively expelled.

\medskip

\noindent \textbf{Effectively Expelled:}
This group consists of 2036 individuals who were transported back to their countries by airplane, boat, or car. The observed frequencies ($n_i$) of individuals apprehended exactly $i$ times are: $n_1 = 1999$, $n_2 = 33$, $n_3 = 2$, $n_4 = 1$, $n_5 = 1$.

\noindent The data exhibit a strongly heavy-tailed pattern, with a high concentration of individuals apprehended only once and rapidly declining frequencies thereafter. To assess whether these frequencies are consistent with a DPareto distribution, we construct a log-log plot of the frequency against the number of apprehensions. As shown in Figure~\ref{PII}, the plot demonstrates a linear, negatively sloped trend—an expected characteristic under a power-law model.

\noindent Theoretically, the slope of the log-log plot is expected to be approximately $-\nu$, where $\nu$ is the exponent of the DPareto distribution. For this dataset, the estimated slope is $-5.02$, and the MLE yields $\hat{\nu} = 5.89$, indicating reasonable agreement between the graphical and inferential estimators.

\noindent To formally evaluate the fit, we apply the proposed test statistic as described in the previous section. The corresponding bootstrap $p$-value (see Table~\ref{SII}) exceeds the 5\% significance level, indicating no significant departure from the null hypothesis. Thus, both visual and statistical evidence support the conclusion that the data for the "Effectively Expelled" group can be reasonably modeled by a DPareto distribution.

\medskip

\noindent \textbf{Not Effectively Expelled:}
This subgroup includes 1880 individuals who were not formally removed from the Netherlands but were released under informal or ineffective conditions—such as being "sent away" or "left with destination unknown"—and hence had the possibility of repeated apprehension within the country. The observed apprehension frequencies are: $n_1 = 1645$, $n_2 = 183$, $n_3 = 37$, $n_4 = 13$, $n_5 = 1$, $n_6 = 1$.

\noindent Although the frequency pattern shows a decreasing trend, the corresponding log-log plot (Figure~\ref{PIII}) appears approximately linear, with an estimated slope of $-4.31$. The MLE of the exponent is $\hat{\nu} = 3.50$, suggesting a rough alignment between the visual and inferential estimates.

\noindent However, it is important to note that an approximately linear log-log plot does not guarantee that the data truly follow a DPareto distribution. To formally assess the fit, we apply the proposed test statistic, which yields a bootstrap $p$-value below the 5\% significance level (see Table~\ref{SIII}). This leads to a statistically significant rejection of the null hypothesis, indicating that the data deviate from the assumptions of the DPareto model.

\noindent This finding highlights a critical point: while a log-log plot may visually suggest power-law behavior, it may mask subtle but systematic deviations from the model that a formal statistical test can detect. Such discrepancies are especially likely when deviations occur in the lower range of $x$, which have a disproportionate effect on the overall fit but may appear visually negligible on a log scale.

\medskip

\noindent 
Our analysis reveals a striking contrast between the two groups. The DPareto distribution provides a good fit for the "Effectively Expelled" individuals, who are less likely to be repeatedly apprehended. In contrast, the "Not Effectively Expelled" group shows statistically significant deviation from the DPareto model despite exhibiting an approximately linear log-log trend. This discrepancy underscores the importance of using formal gof tests in conjunction with graphical tools and highlights the influence of behavioral, procedural, or enforcement-related heterogeneity in shaping apprehension patterns. These findings demonstrate the utility of the proposed test in distinguishing between true and apparent adherence to power-law behavior in count data.

\begin{figure}[H]
\centering
\begin{minipage}[t]{0.495\textwidth}
  \centering
  \includegraphics[width=\linewidth]{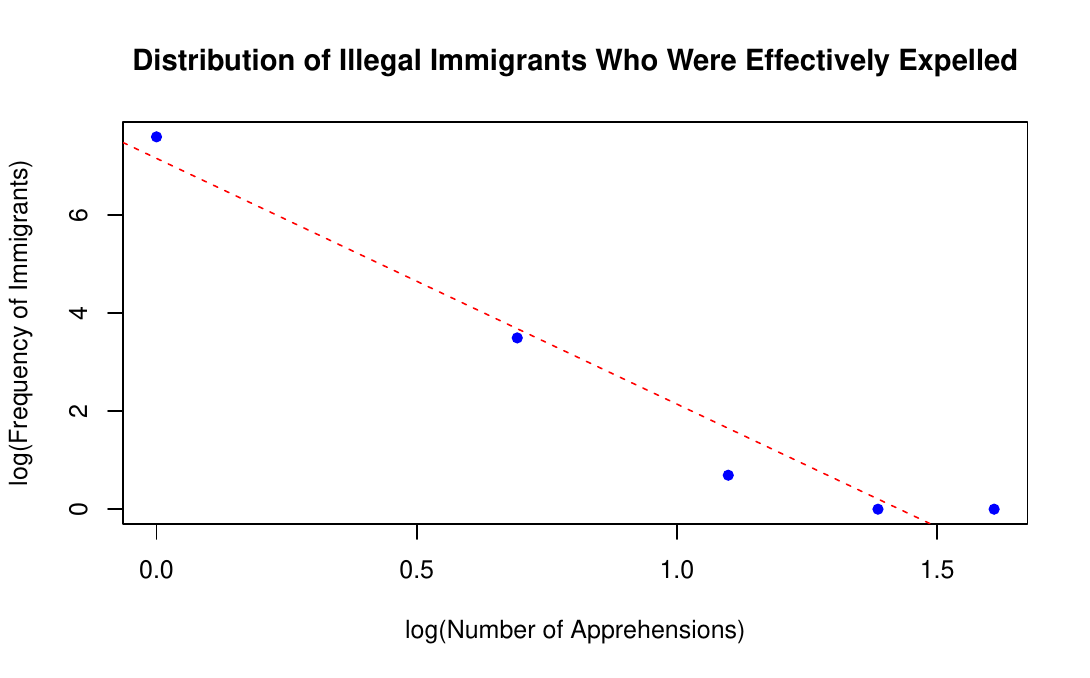}
  \caption{Log-Log plot of Illegal Immigrants who were Effectively Expelled.}
  \label{PII}
\end{minipage}
\hfill
\begin{minipage}[t]{0.475\textwidth}
  \centering
  \includegraphics[width=\linewidth]{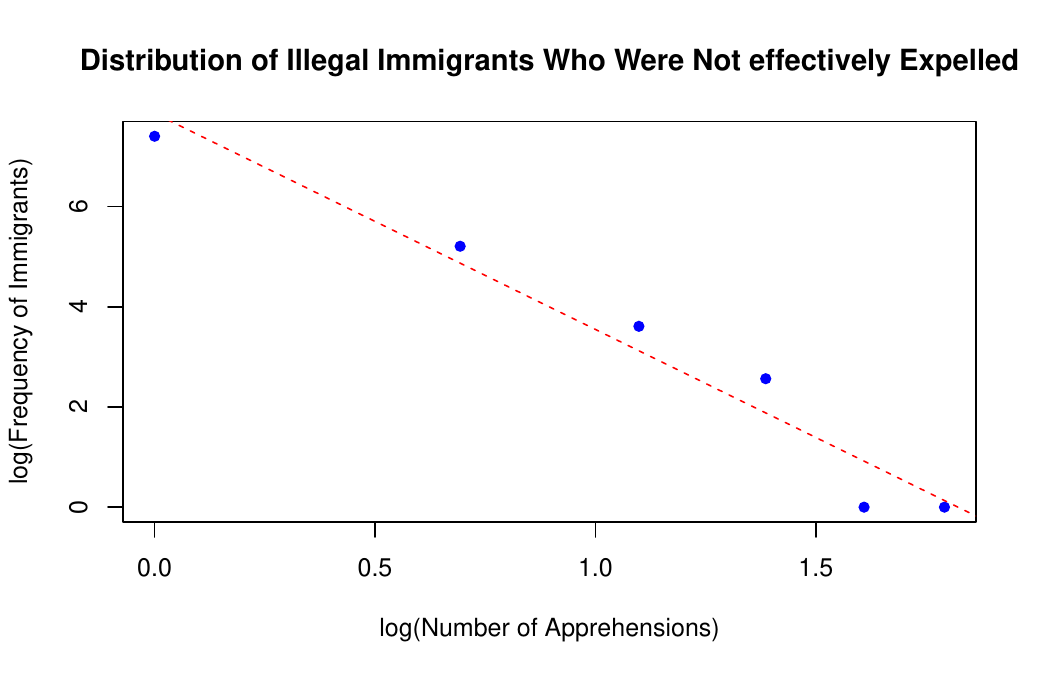}
  \caption{Log-Log plot of Illegal Immigrants who were Not Effectively Expelled.}
  \label{PIII}
\end{minipage}
\end{figure}

\begin{table}[t]
\centering
\begin{minipage}[t]{0.49\textwidth}
  \centering
  \begin{tabular}{ccc} \hline
    Tests & TS Value & p-value \\ \hline
    $\mathfrak{K}$ & $5.05 \times 10^{-6}$ & 0.42 \\
    $Z_{n,0.5}$ & $1.81 \times 10^{-4}$ & 0.53 \\
    $Z_{n,1}$ & $3.99 \times 10^{-5}$ & 0.52 \\
    $Z_{n,2}$ & $5.80 \times 10^{-6}$ & 0.49 \\ 
    $T_{n,0}$  & $4.49 \times 10^{-2}$ & 0.75  \\
    $T_{n,1}$  & $2.21 \times 10^{-2}$ & 0.76  \\
    $T_{n,2}$  & $1.24 \times 10^{-2}$ & 0.76   \\
    $T_{n,3}$  & $7.67 \times 10^{-3}$ & 0.76   \\
    $T_{n,4}$  & $5.06 \times 10^{-3}$ & 0.77   \\
    $T_{n,5}$  & $3.51 \times 10^{-3}$ & 0.77   \\
    $S_{BEN} $ & $4.81 \times 10^{-7}$ & 0.63   \\
    $C_n^e$ & $2.57 \times 10^{-4}$ & 0.57 \\  \hline
  \end{tabular}
  \vspace{0.2cm}
  \caption{Summary of results for the illegal immigrants who were effectively expelled}
  \label{SII}
\end{minipage}
\hfill
\begin{minipage}[t]{0.49\textwidth}
  \centering
  \begin{tabular}{ccc} \hline
    Tests & TS Value & p-value \\ \hline
    $\mathfrak{K}$ & $0.42$ & $<0.001$ \\
    $Z_{n,0.5}$ & $0.53$ & $<0.001$ \\
    $Z_{n,1}$ & $0.51$ & $<0.001$ \\
    $Z_{n,2}$ & $0.49$ & $<0.001$ \\ 
    $T_{n,0}$  & 2.93 & $<0.001$    \\
    $T_{n,1}$  & 1.49  & $<0.001$    \\
    $T_{n,2}$  & 0.86  & $<0.001$   \\
    $T_{n,3}$ & 0.54   & $<0.001$    \\
    $T_{n,4}$ & 0.36   & $<0.001$    \\
    $T_{n,5}$ & 0.25   & $<0.001$   \\
    $S_{BEN}$ & $2.35 \times 10^{-4}$ & $<0.001$    \\
    $C_n^e$ & 0.26     & $<0.001$  \\ \hline
  \end{tabular}
  \vspace{0.2cm}
  \caption{Summary of results for the illegal immigrants who were not effectively expelled}
  \label{SIII}
\end{minipage}
\end{table}

\section{Conclusion and Open Problems} \label{conclusion}

In this paper, we proposed a novel gof test for the DPareto distribution based on a new Stein characterization. The test statistic was constructed by choosing a specific Stein class of test functions $\{f(x) = 1 - s^{x-1}: \;s\in(0,1)\}$ which enabled us to derive a tractable expression, derive the thoery of the testing procedure and evaluate its performance. The finite-sample behavior of the proposed test was examined through extensive Monte Carlo simulations, demonstrating its competitive power in comparison with existing test. Furthermore, the practical usefulness of the test was illustrated using two real-world data sets.

\medskip

\noindent
The work presented in this paper opens up several natural directions for further research. One possible extension is to investigate alternative Stein classes of test functions $f$. For instance, other simple polynomial forms, characteristic function-type or piecewise-defined functions may offer improved power against specific alternatives to the DPareto distribution. The proposed test could also be adapted to handle real-world scenarios where data are censored or truncated-a common occurrence in domains such as income, word frequency analysis, or web traffic, which often exhibit DPareto-like behavior. Additionally, extending the Stein characterization approach to other discrete heavy-tailed or power-law distributions, such as the Zipf--Mandelbrot or Yule--Simon distributions, would broaden the applicability of the proposed methodology. Furthermore, another interesting open problem is deriving explicit expressions for the eigenvalues of the covariance operator 
$T$ in Theorem \ref{thm:null-limit}. While a closed-form analytical treatment may be infeasible, the numerical methods developed in \cite{EJM:2025} should be applicable. Moreover, since these eigenvalues depend on the unknown parameter $\nu_0$, accurate and fast implementation of case by case approximation could obviate the need for a bootstrap procedure \citep[Section 5]{EJM:2025}. A detailed local power (in the sense of contiguous alternatives) analysis is beyond the scope of this paper, but could be carried out by perturbing the DPareto pmf and using the Hilbert-space methodology. We expect the test to be locally asymptotically optimal in some directions, but this remains as an open research question.

\section{Disclosure statement} No potential conflict of interest was reported by the authors.
\section{Data Availability Statement} The data used in this study are derived from previously published sources and are included within this article. 
\section{Acknowledgements} 
Sakshi Khandelwal would like to express sincere gratitude to the Department of Science and Technology (DST) for the generous support provided through the Inspire Fellowship  (Inspire code: IF220245).

\bibliographystyle{apalike} 
\bibliography{bibtex}

\begin{appendix}
\section{Proofs}\label{app:Vstat-proofs}
\begin{proof}[Proof of Theorem~\ref{chth}]
The proof of the necessary part follows directly by obtaining the right-hand side expectation for the DPareto distribution given as
\begin{align*}
\mathbb{E} \bigg( \bigg(1- \bigg(\frac{X}{X+1}\bigg)^\nu \bigg) \, s \, (1-s^X)\bigg)
&= \mathbb{E}\left(s(1-s^X\right))-\mathbb{E}\left(s(1-s^X)\left(\frac{X}{X+1}\right)^\nu\right)\\
&= s-s\frac{\text{Li}_\nu(s)}{\zeta(\nu)}-\frac{s}{\zeta(\nu)}(\zeta(\nu)-1)+\frac{1}{\nu}\sum_{k=1}^{\infty} \frac{s^{k+1}}{(k+1)^\nu} \\
&=s-s\frac{\text{Li}_\nu(s)}{\zeta(\nu)}-s+\frac{s}{\zeta(\nu)}+\frac{\text{Li}_\nu(s)}{\zeta(\nu)}-\frac{s}{\zeta(\nu)} \\
&=(1-s)\frac{\text{Li}_\nu(s)}{\zeta(\nu)}=(1-s)G_X(s).
\end{align*}
which is required left-hand side. Now, to prove the sufficient part of the theorem, suppose that the random variable $X$ satisfies the functional equation (\ref{eq:NullStein}), which characterize the DPareto distribution. To prove this, let $p_X(x)$ be any pmf on $\mathbb{N}$ satisfying equation \eqref{eq:NullStein} for all $s\in (0,1)$. Utilizing the definition of pgf $G_X(s)=\sum_{k=1}^{\infty}s^k p_X(k)$ in (\ref{eq:NullStein}) and matching the coefficient of $s^k$ for $k\ge 2$ we get,
\begin{equation*}
p_X(k)-p_X(k-1)= -\left(1-\left(\frac{k-1}{k}\right)^\nu \right) p_X(k-1)
\end{equation*}
reduces to 
\begin{equation*}
p_X(k)=\left(\frac{k-1}{k}\right)^\nu p_X(k-1)=p_X(1) \prod_{j=2}^{k} \left(\frac{j-1}{j} \right)^\nu =p_X(1)\frac{1}{k^\nu}.
\end{equation*} 
Normalization yields $p_X(1)=1/\zeta(\nu)$ and $p_X(k)=(k^\nu \zeta(\nu))^{-1}$. Thus the family $p_X(x)$ uniquely characterizes the DPareto law. 
\end{proof}
\noindent The following Lemmas are used in the proof of Theorem \ref{thm:null-limit}.
\begin{lem}\label{lem:bounds-app}
Fix a compact interval $I\subset(1,\infty)$. Then, for all $\nu\in I$, $x\in\mathbb{N}$ and $s\in(0,1)$,
\[
|g_\nu(x,s)|\le 3,\qquad 
|\partial_\nu g_\nu(x,s)|\le x^{-1},\qquad
|\partial_{\nu\nu}^2 g_\nu(x,s)|\le x^{-2},
\]
and
\[
\partial_\nu g_\nu(x,s)=\Big(\frac{x}{x+1}\Big)^\nu \log\Big(\frac{x}{x+1}\Big)\,(1-s^x),\qquad
\partial_{\nu\nu}^2 g_\nu(x,s)=\Big(\frac{x}{x+1}\Big)^\nu \Big(\log\Big(\frac{x}{x+1}\Big)\Big)^2\,(1-s^x).
\]
\end{lem}

\begin{proof}
The expressions follow by differentiation. The bounds use $0<s^x\le 1$, $0<s(1-s^x)\le 1$,
$0\le (x/(x+1))^\nu\le 1$, and $\log(1+1/x)\le 1/x$.
\end{proof}

\begin{lem}\label{lem:kernel-regularity-app}
The kernel $k$ satisfies $k\in L^2(P_{\nu_0}\times P_{\nu_0})$ and is positive semidefinite. Hence the operator
$T$ defined in Theorem~\ref{thm:null-limit} is self-adjoint, nonnegative and Hilbert--Schmidt.
\end{lem}

\begin{proof}
Square-integrability follows from Cauchy--Schwarz
\begin{equation*}
    |k_{\nu_0}(x,y)|^2\le \|\psi_{\nu_0}(x,\cdot)\|_{L^2(0,1)}^2\|\psi_{\nu_0}(y,\cdot)\|_{L^2(0,1)}^2
\end{equation*}
and $\mathbb{E}\|\psi_{\nu_0}(X,\cdot)\|_{L^2}^2<\infty$ (since $g_{\nu_0}$ is bounded, $a\in L^2$, $\mathbb{E}\left(\ell(X,\nu_0)^2\right)<\infty$).
Positive semidefiniteness holds since
$\sum_{i,j}c_ic_jk(x_i,x_j)=\int_0^1(\sum_i c_i\psi_{\nu_0}(x_i,s))^2ds\ge 0$.
\end{proof}

\begin{proof}[Proof of Theorem~\ref{thm:null-limit}]
We apply the theory given in \citet[Proposition 1]{LN:2013}. Under the assumptions of the section, $(X_i)_{i\ge1}$ is iid\ with law $P_{\nu_0}$. Thus, the dependence assumptions in (A2) (i) of 
\citet[Section~3]{LN:2013} hold trivially in the iid\ case. Symmetry is immediate from $h_\nu(x,y)=h_\nu(y,x)$.
For nonnegative definiteness, for any $m\in\mathbb{N}$, $c_1,\dots,c_m\in\mathbb{R}$ and $x_1,\dots,x_m\in\mathbb{N}$,
\[
\sum_{i,j=1}^m c_i c_j h_{\nu}(x_i,x_j)
=\int_0^1 \Big(\sum_{i=1}^m c_i g_\nu(x_i,s)\Big)^2 ds\ge 0,
\]
so $h_\nu$ is positive semidefinite and hence satisfies (A1) (ii). Next, using the uniform bound $|g_\nu(x,s)|\le 3$ from Lemma~\ref{lem:bounds-app},
\[
0\le h_{\nu_0}(x,x)=\int_0^1 g_{\nu_0}(x,s)^2\,ds\le9,
\]
hence $\mathbb{E}_{\nu_0}[h_{\nu_0}(X,X)]\le 9<\infty$ hence (A1) (iii) is satisfied. Finally, by the assumed Stein identity
\eqref{eq:Stein_id}, for any $x\in\mathbb{N}$,
\[
\mathbb{E}_{\nu_0}\!\big[h_{\nu_0}(x,X)\big]
=\int_0^1 g_{\nu_0}(x,s)\,\mathbb{E}_{\nu_0}[g_{\nu_0}(X,s)]\,ds
=0.
\]
Thus $h_{\nu_0}$ is $1$-degenerate under $P_{\nu_0}$ in the sense required by (A1) (iv).

The only $\nu$-dependence in $g_\nu(x,s)$ occurs through $\big(x/(x+1)\big)^\nu$.
Lemma~\ref{lem:bounds-app} yields the derivative bounds
$|\partial_\nu g_\nu(x,s)|\le x^{-1}$ and $|\partial_{\nu\nu}^2 g_\nu(x,s)|\le x^{-2}$ uniformly for $\nu$ in a compact
neighborhood of $\nu_0$ and for all $s\in(0,1)$.
In particular, the map $\nu\mapsto g_\nu(x,\cdot)$ is twice Fr\'echet differentiable as an element of $L^2(0,1)$,
with derivatives dominated by integrable bounds under $P_{\nu_0}$ because
$\mathbb{E}_{\nu_0}[X^{-2}]<\infty$ for $\nu_0>1$.

Moreover, the estimator $\widehat\nu_n$ satisfies the linear representation \eqref{eq:linear-rep}
with $\mathbb{E}_{\nu_0}[\ell(X,\nu_0)]=0$ and $\mathbb{E}_{\nu_0}[\ell(X,\nu_0)^2]<\infty$ by assumption. By Lemma \ref{lem:kernel-regularity-app}, $k\in L^2(P_{\nu_0}\times P_{\nu_0})$ and $k$ is positive semidefinite. The differentiability bounds together with \eqref{eq:linear-rep} imply the $L^2(0,1)$ linearization
\[
\sqrt{n}\,\bar g_n(\cdot;\widehat\nu_n)
=\frac{1}{\sqrt{n}}\sum_{i=1}^n \psi_{\nu_0}(X_i,\cdot)+o_{\mathbb{P}}(1)
\qquad\text{in }L^2(0,1),
\]
where $\bar g_n(s;\nu)=n^{-1}\sum_{i=1}^n g_\nu(X_i,s)$, $s\in(0,1)$.
Taking squared $L^2$-norms and expanding yields
\[
\mathfrak{K}
=n\|\bar g_n(\cdot;\widehat\nu_n)\|_{L^2(0,1)}^2
=\frac1n\sum_{i,j=1}^n k(X_i,X_j)+o_{\mathbb{P}}(1)
=nV_n(k)+o_{\mathbb{P}}(1).
\]
This is precisely the reduction asserted by \citet[Proposition~1]{LN:2013} for estimated-parameter degenerate V-statistics in our iid\ setting.

Under \eqref{eq:Stein_id} and $\mathbb{E}_{\nu_0}[\ell(X,\nu_0)]=0$ we have
$\mathbb{E}_{\nu_0}[\psi_{\nu_0}(X,s)]=0$ for all $s$, hence $k$ is $1$-degenerate
\[
\mathbb{E}_{\nu_0}[k(x,X)]
=\int_0^1 \psi_{\nu_0}(x,s)\,\mathbb{E}_{\nu_0}[\psi_{\nu_0}(X,s)]\,ds
=0,\qquad \forall x\in\mathbb{N}.
\]
Let $T:L^2(P_{\nu_0})\to L^2(P_{\nu_0})$ be the Hilbert--Schmidt operator
\[
(Tf)(x):=\int_{\mathbb{N}} k_{\nu_0}(x,y)f(y)\,dP_{\nu_0}(y).
\]
Since $k$ is symmetric, positive semidefinite and belongs to $L^2(P_{\nu_0}\times P_{\nu_0})$, $T$ is compact, self-adjoint and nonnegative. Denote by $\{\lambda_m\}_{m\ge1}$ its decreasing positive sequence of eigenvalues. The degenerate V-statistic limit theorem \cite[Theorem~1]{LN:2013}
then yields
\[
nV_n(k)\ \stackrel{d}{\rightarrow} \sum_{m=1}^\infty \lambda_m Z_m^2,
\]
where $Z_1,Z_2,\dots$ are iid\ $N(0,1)$. Since $\mathfrak{K}=nV_n(k)+o_{\mathbb{P}}(1)$, Slutsky's theorem yields the stated result.
\end{proof}
\end{appendix}
\end{document}